\documentclass[11pt]{article}

\usepackage[a4paper,margin=1in]{geometry}
\usepackage[T1]{fontenc}
\usepackage[utf8]{inputenc}
\usepackage{lmodern}
\usepackage{microtype}
\usepackage{amsmath,amssymb,amsfonts,mathtools}
\usepackage{amsthm}
\usepackage{bm}
\usepackage{booktabs}
\usepackage{pgfplots}
\usepackage{hyperref}   
\usepackage{cleveref}  
\pgfplotsset{compat=1.18}

\title{A Thermodynamic SU(1,1) Witness Framework for Double-Quantum NMR Signals in Neural Tissue}

\author{Christian Kerskens \\
\small Trinity College Institute of Neuroscience}
\date{\today}

\newtheorem{theorem}{Theorem}
\newtheorem{remark}{Remark}

\begin{document}

\maketitle

\begin{abstract}
Entanglement criteria based on variances or Fisher information are well developed for compact collective spin algebras, but their extension to non-compact dynamical sectors is less straightforward. In particular, double-quantum (DQ) observables associated with effective SU(1,1) structures can lead to formally unbounded classical fluctuation estimates unless additional physical constraints are imposed.

In this note, we develop a thermodynamic witness framework in which the classically accessible fluctuation sector is strictly bounded by finite-temperature detailed-balance conditions and motionally narrowed sequence-transfer limits. By analyzing the quantum dynamical semigroup of the spin-bath interaction, we demonstrate that spontaneous transient pair correlations generated by a stationary incoherent bath are contractively capped near an amplitude of \(10^{-9}\). Furthermore, classical coherent sequence amplification is empirically bounded to \(\mathcal{O}(10^{-2})\) in motionally narrowed tissue. 

The resulting functional provides a concrete, theoretically derived bounding framework against which macroscopic DQ anomalies (e.g., fractional amplitudes on the order of \(10\%\) to \(15\%\)) can be rigorously classified as classically inexplicable, provided macro-scale structural stability (constant \(T_2^*\)) is empirically verified.
\end{abstract}

\section{Introduction}
Entanglement witnesses based on collective variances or quantum Fisher information are standard tools in compact spin systems \cite{Pezze2018}, where the relevant generators typically belong to SU(2)-type algebras with well-controlled separable bounds. In double-quantum (DQ) NMR settings, however, the effective dynamical variables may instead organize into a non-compact SU(1,1)-like algebra. In such cases, naive extensions of compact-algebra witness bounds can become ill-posed unless the physically accessible excitation sector is constrained.

The present note proposes a concrete framework in which a finite witness threshold is obtained by supplementing the spin-sector description with explicit open-quantum-system constraints. The central idea is that if the biological medium acts as a finite-temperature bath, the classical fluctuations and sequence-driven transfers available to generate macroscopic DQ processes must be strictly bounded.

We decompose the classical baseline into two mechanisms: (1) spontaneous excitation by an incoherent thermal bath, bounded by detailed balance, and (2) coherent sequence-driven transfer, bounded by the tissue's motional narrowing. We argue that observables routinely used to monitor physiological stability---such as \(T_2^*\) and single-quantum coherence (SQC)---provide empirical constraints that seal potential macroscopic loopholes.

\section{Algebraic Classification: Compact and Non-Compact Two-Spin Sectors}
\label{sec:algebra}
The two-spin-\(\tfrac12\) operator algebra \(\mathfrak{su}(4)\)
decomposes by coherence order into dynamically closed subalgebras. This
classification is essential both for the witness construction developed
here and for the companion covariance-geometric analysis of the same
signal~\cite{Paper3}.

\subsection{Zero-quantum sector: compact SU(2)}

The zero-quantum (ZQ) sector (\(p=0\)) is spanned by
\[
  \mathcal{B}_{\mathrm{ZQ}}
  = \{I_{1z},\; I_{2z},\; I_{1+}I_{2-},\; I_{1-}I_{2+}\}.
\]
Within this space, the compact algebra
\begin{equation}
  \mathfrak{su}(2)_{\mathrm{ZQ}}:\quad
  \left\{I_{1+}I_{2-},\; I_{1-}I_{2+},\;
    \tfrac12(I_{1z}-I_{2z})\right\}
  \label{eq:su2ZQ}
\end{equation}
closes under commutation and generates bounded oscillatory exchange
dynamics at the difference frequency
\(\Omega_- = \omega_1 - \omega_2\)~\cite{Ernst1990}.

\subsection{Double-quantum sector: non-compact SU(1,1)}

The double-quantum (DQ) sector (\(|p|=2\)) contains the pair operators
\(\{I_{1+}I_{2+},\; I_{1-}I_{2-}\}\). Together with the diagonal
generator \(\tfrac12(I_{1z}+I_{2z})\), these close to form the
non-compact algebra
\begin{equation}
  \mathfrak{su}(1,1)_{\mathrm{DQ}}:\quad
  \left\{K_+ = I_{1+}I_{2+},\; K_- = I_{1-}I_{2-},\;
    K_0 = \tfrac12(I_{1z}+I_{2z})\right\},
  \label{eq:su11DQ}
\end{equation}
satisfying
\begin{equation}
  [K_0, K_\pm] = \pm K_\pm,
  \qquad
  [K_-, K_+] = 2K_0.
  \label{eq:su11comm}
\end{equation}
The generators \(K_\pm\) carry coherence order \(p = \pm 2\) and
connect the aligned states \( |\!\uparrow\uparrow\rangle \) and
\( |\!\downarrow\downarrow\rangle \).

\subsection{Dynamical distinction}

The crucial distinction lies in the geometry of the dynamics. Compact
\(\mathfrak{su}(2)\) possesses a positive-definite Killing form and
generates purely imaginary adjoint eigenvalues, leading to bounded
oscillatory evolution. The non-compact algebra \(\mathfrak{su}(1,1)\)
has indefinite Killing-form signature, admits real adjoint eigenvalues,
and supports hyperbolic trajectories
\cite{Perelomov1986,GerryKnight2004}:
\begin{equation}
  S_{\mathrm{SU(2)}}(t) \propto \sin(\Omega_- t),
  \qquad
  S_{\mathrm{SU(1,1)}}(t) \propto \sinh(g\, t).
  \label{eq:su2su11signal}
\end{equation}
A proof that compact algebras cannot produce hyperbolic growth is given
in \Cref{app:hyperbolicity}.

\begin{remark}[Unbounded quantum sector]
\label{rem:unbounded}
Unlike in the compact SU(2) case, the non-compact SU(1,1)
representation space admits coherent states with arbitrarily large
pair-correlation amplitudes~\cite{Perelomov1986}. The quantum sector is
therefore not subject to an analogous ceiling. This ensures that the
witness construction is not vacuous: the classical bound derived below
excludes classical mechanisms, while the quantum sector can in principle
accommodate the macroscopic signals observed.
\end{remark}

\subsection{The need for regularization}

Classical separability imposes inequalities of the form
\(S_{\mathrm{DQ}} \leq \mathcal{C}\,\langle K_0 \rangle^2\), where
\(\mathcal{C}\) is an algebraic factor and \(S_{\mathrm{DQ}}\) is the
measured DQ response. Because SU(1,1) is non-compact, \(\langle
K_0\rangle\) is not bounded by the algebra alone. The resulting
inequality ceases to be useful as a witness threshold unless the
physically accessible excitation sector is constrained by additional
physical assumptions.

To eliminate symbolic abstraction, it is experimentally robust to
express the witness in terms of the measurable fractional DQ amplitude
\(f_{\mathrm{DQ}} = S_{\mathrm{DQ}} / M_0\), calibrated against the
absolute thermal-equilibrium single-quantum magnetization \(M_0\).
Bounding \(f_{\mathrm{DQ}}\) requires identifying the physical limits
on classically accessible pair-correlation generation.

\section{Classical Limits: Incoherent Bath vs. Coherent Transfer}
To address the dynamics of transient signals, we must distinguish between spontaneous transient excitation by the thermal bath and coherent amplification by the applied RF pulse sequence.

\paragraph{1. The Incoherent Bath Limit (\(\epsilon_{\mathrm{th}}\)).}
Between coherent RF pulses, the spin system \(\rho_S\) evolves under the local residual dipolar coupling and the surrounding thermal bath. Under the standard Born-Markov approximation, this open-system dynamics is governed by a quantum dynamical semigroup, typically expressed as a Lindblad or Redfield master equation \cite{Breuer2002}. 

For a stationary thermal bath, the dissipator strictly satisfies the Kubo-Martin-Schwinger (KMS) detailed-balance condition, establishing the thermal Gibbs state \(\rho_{\mathrm{th}}\) as the unique asymptotic fixed point. A fundamental property of completely positive trace-preserving (CPTP) semigroups is the contractivity of the relative entropy \cite{Spohn1978}. Because the distance between the instantaneous state \(\rho_S(t)\) and \(\rho_{\mathrm{th}}\) decreases monotonically, an incoherent stationary bath cannot transiently pump pair correlations from a less-correlated state to a level exceeding the thermal equilibrium baseline. Any transient excursion from a baseline state approaches the thermal ceiling strictly from below.

The thermodynamic ceiling on classically spontaneous pair-correlations is therefore governed by the high-temperature expansion parameter:
\begin{equation}
    \epsilon_{\mathrm{th}} = \frac{\hbar \omega_D}{kT}.
\end{equation}
For restricted water protons in neural tissue (\(T \approx 310\) K), assuming a maximal local dipolar fluctuation amplitude of \(\omega_D/2\pi \approx 10\) kHz, the interaction energy is \(\hbar \omega_D \approx 6.6 \times 10^{-30}\) J. Consequently, \(\epsilon_{\mathrm{th}} \approx 1.5 \times 10^{-9}\). Contractivity dictates that a stationary classical bath cannot transiently force the spin system into a correlated state exceeding \(\mathcal{O}(10^{-9})\).

\paragraph{2. Coherent Sequence Amplification (\(\eta_{\mathrm{seq}}\)).}
A transient DQ response could be driven classically if the MQC pulse sequence coherently transfers existing Zeeman magnetization into multiple-quantum order. In solid-state NMR with rigid structural networks \cite{Baum1985}, pulse sequences efficiently convert thermal order into macroscopic MQC signals.

Such classical sequence amplification relies intrinsically on strong, \textit{static} (time-averaged) dipolar couplings. In neural tissue, extreme motional narrowing attenuates the residual static coupling \(\overline{\omega}_D\). For water in highly anisotropic environments (e.g., myelin), \(\overline{\omega}_D/2\pi\) is at most a few Hz. Given typical mixing times \(t_m \sim 5\) ms, the fractional sequence transfer scales roughly as:
\begin{equation}
    \eta_{\mathrm{seq}} \sim (\overline{\omega}_D t_m)^2 \approx (2\pi \times 5 \text{ s}^{-1} \times 0.005 \text{ s})^2 \approx 2.5 \times 10^{-2}.
\end{equation}
This estimate is not intended as a universal theorem on classical
sequence transfer. The scaling $(\overline{\omega}_D t_m)^2$ is an
order-of-magnitude estimate for short-time coherent transfer under weak
static couplings; the specific values of $\overline{\omega}_D$ and
$t_m$ are chosen to be generous upper estimates for the relevant tissue
regime. The precise numerical prefactor is therefore open to reasonable
dispute. The point, however, is structural: even under deliberately
optimistic classical assumptions, the resulting transfer ceiling remains
an order of magnitude below the reported macroscopic signal amplitude of
$0.15$. The argument would only fail if $\overline{\omega}_D$ were
larger by an order of magnitude---but that would require a
rigid-lattice-like coupling regime incompatible with the observed stable
$T_2^*$.

\section{The Role of the Bath Spectral Density}
While detailed balance imposes the asymptotic ceiling (\(\epsilon_{\mathrm{th}}\)), the timescale required to approach this limit is governed by the bath's spectral density. Assuming an exponentially decaying local fluctuating field with correlation time \(\tau_c\), the standard Bloembergen--Purcell--Pound (BPP) model \cite{Bloembergen1948, Abragam1961} dictates a Lorentzian spectral profile:
\begin{equation}
    J(\omega_0) = \frac{2 \langle \omega_D^2 \rangle \tau_c}{1 + \omega_0^2 \tau_c^2}.
\end{equation}
This profile is illustrated in Figure~\ref{fig:bpp_spectral_density}. For bulk cerebrospinal fluid (CSF), \(\tau_c\) is in the picosecond regime. Extreme motional narrowing (\(\omega_0 \tau_c \ll 1\)) collapses the spectral density, rendering the bath dynamically incapable of driving transitions at the requisite NMR frequencies. For restricted myelin water (\(\tau_c \sim 1\) ns), \(J(\omega_0)\) approaches its maximum. However, it is crucial to recognize that an optimal \(J(\omega_0)\) merely maximizes the \textit{rate} of equilibration; it does not alter the thermodynamic ceiling. The pair-correlation ceiling remains firmly bounded by \(\epsilon_{\mathrm{th}}\).

\begin{figure}[t]
\centering
\begin{tikzpicture}
\begin{axis}[
    width=0.75\textwidth,
    height=0.45\textwidth,
    domain=0:5,
    samples=300,
    xlabel={$x=\omega_0 \tau_c$},
    ylabel={normalized spectral density $\tilde{J}(x)$},
    ymin=0, ymax=1.1,
    xmin=0, xmax=5,
    xtick={0,1,2,3,4,5},
    ytick={0,0.5,1},
    axis lines=left,
    enlargelimits=false,
    clip=false
]
\addplot[thick] {2*x/(1+x^2)};
\addplot[only marks, mark=*, mark size=1.8pt] coordinates {(1,1)};
\node[anchor=west, align=left] at (axis cs:0.18,0.18) {\small fast-motion regime};
\node[anchor=west, align=left] at (axis cs:1.15,1.02) {\small maximal rate (restricted)};
\node[anchor=west, align=left] at (axis cs:3.0,0.55) {\small slow/off-resonant};
\draw[->] (axis cs:0.35,0.32) -- (axis cs:0.08,0.08);
\node[anchor=west] at (axis cs:0.42,0.34) {\small bulk CSF};
\end{axis}
\end{tikzpicture}
\caption{
Normalized spectral density \(\tilde{J}(x)=\frac{2x}{1+x^2}\). The bath spectral density governs the transition \textit{rate} toward equilibrium, but does not alter the thermodynamic \(\epsilon_{\mathrm{th}}\) ceiling enforced by detailed balance.
}
\label{fig:bpp_spectral_density}
\end{figure}
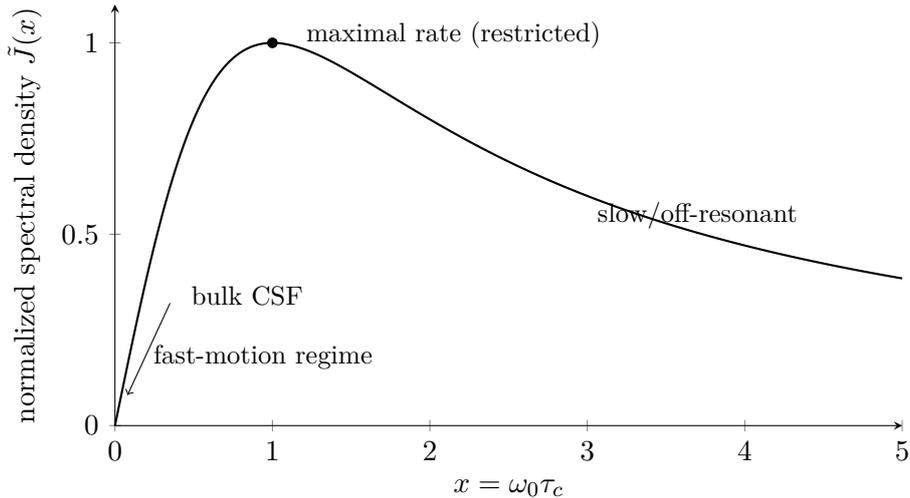

\section{Empirical Constraints and Non-Stationary Baths}
A potential theoretical loophole is that the biological event triggering the DQ burst renders the bath transiently \textit{non-stationary}. A non-stationary bath violates detailed balance and could, in principle, drive the system away from the Gibbs state. 

However, macroscopic observables provide an inescapable empirical constraint. Single-quantum transverse relaxation (\(T_2^*\)) is exquisitely sensitive to both microscopic bath dynamics and macroscopic static dephasing \cite{Ogawa1990}. If a tissue compartment were to experience a structural shift severe enough to create a strongly non-stationary local bath or massively increase \(\overline{\omega}_D\), this sudden restriction in field homogeneity would cause transient line-broadening, resulting in a precipitous drop in \(T_2^*\).

The simultaneous observation of a macroscopic DQ burst alongside a stable, unperturbed \(T_2^*\) provides strict empirical evidence that the underlying dipolar network remained stationary and motionally narrowed. A stable \(T_2^*\) explicitly caps the non-stationary loophole and restricts the coherent transfer parameter \(\eta_{\mathrm{seq}}\) to its negligible baseline. Similarly, tracking Magnetization Transfer (MT) contrast \cite{Henkelman2001} alongside SQC observables prevents loopholes regarding abrupt physiological exchange events between the visible water pool and restricted macromolecular environments.

\section{Definition of the Thermodynamic Witness}
We define the classical bounding threshold directly on the fractional DQ response. The maximum classically accessible fractional signal is bounded by the sum of spontaneous bath generation and coherent sequence transfer:
\begin{equation}
    f_{\mathrm{class}}^{\max} = \epsilon_{\mathrm{th}} + \eta_{\mathrm{seq}}(T_2^*),
\end{equation}
where \(\epsilon_{\mathrm{th}} \sim \mathcal{O}(10^{-9})\) is the detailed-balance limit, and \(\eta_{\mathrm{seq}} \sim \mathcal{O}(10^{-2})\) is the sequence transfer efficiency, empirically capped by the stability of \(T_2^*\). The measurable thermodynamic SU(1,1) witness is defined as:
\begin{equation}
    \mathcal{W}_{\mathrm{th}} = f_{\mathrm{DQ}}^{\mathrm{measured}} - f_{\mathrm{class}}^{\max}.
\end{equation}

\begin{theorem}
For a spin system coupled to a stationary incoherent thermal bath and
subjected to an RF sequence under stable macroscopic dephasing
(constant \(T_2^*\)), the maximal classical fractional DQ amplitude
within the present model class is bounded by \(f_{\mathrm{class}}^{\max}\).
An experimental observation yielding \(\mathcal{W}_{\mathrm{th}} > 0\)
therefore excludes this restricted class of classical stationary
spin-bath and motionally narrowed sequence-transfer mechanisms.
\end{theorem}

\begin{proof}
By the contractivity of the quantum dynamical semigroup \cite{Spohn1978}, spontaneous transitions driven by a stationary, incoherent bath satisfying the KMS condition cannot overshoot the thermal equilibrium parameter \(\epsilon_{\mathrm{th}}\). Secondly, classical coherent transfer of Zeeman order requires static dipolar couplings \(\overline{\omega}_D\). The constancy of \(T_2^*\) serves as an empirical boundary condition guaranteeing that \(\overline{\omega}_D\) remains heavily motionally narrowed, restricting \(\eta_{\mathrm{seq}}\) to \(\mathcal{O}(10^{-2})\). Consequently, any macroscopic response \(f_{\mathrm{DQ}} \gg \epsilon_{\mathrm{th}} + \eta_{\mathrm{seq}}\) violates the maximal classical capacities of the system.
\end{proof}

\section{Discussion and Outlook}
The framework proposed here shows that SU(1,1)-type witness constructions need not remain formally unbounded. Once the driving background is restricted by detailed-balance constraints (\(\sim 10^{-9}\)) and sequence transfer is limited by empirically verified motional narrowing (\(\sim 10^{-2}\)), a calculable threshold emerges. 

If an experimentally observed transient DQ response reaches fractional amplitudes on the order of \(15\%\) \cite{Kerskens2022} under stable macroscopic conditions, it shatters the \(f_{\mathrm{class}}^{\max}\) boundary. Such a massively positive witness places insurmountable strain on standard classical open-quantum-system models.

A related open question is whether transient DQ bursts may be further constrained by finite-rate transport limits. In a broader setting, a Bures--Wasserstein-type information geometry \cite{Bhatia2019} may provide a natural language for quantifying the covariance-transport cost associated with accessing a transient boundary-limited regime of the embodied substrate \cite{Deffner2017}. The companion paper~\cite{Paper3} develops a covariance-geometric route from Bures--Wasserstein boundary dynamics to an observable pair-sector signal in the same DQ/SU(1,1) sector bounded here. Exploring the relation between the thermodynamic limits derived here and such covariance-geometric descriptions is a natural direction for future work.

\appendix

\section{Compact algebras cannot produce hyperbolic growth}
\label{app:hyperbolicity}

Under compact \(\mathfrak{su}(2)_{\mathrm{ZQ}}\) evolution generated by
\(H = J(I_{1+}I_{2-}+I_{1-}I_{2+})\), any expectation value
\(\langle O(t)\rangle\) with
\(O \in \mathfrak{su}(2)_{\mathrm{ZQ}}\) satisfies
\[
  \langle O(t)\rangle = \sum_k c_k e^{i\lambda_k t},
\]
where the adjoint eigenvalues \(\lambda_k\) are purely imaginary.
Hence the evolution is a finite sum of oscillatory terms; hyperbolic
growth is impossible within the compact algebra.

For \(\mathfrak{su}(1,1)\), the Killing form has indefinite signature,
the adjoint representation admits real eigenvalues, and the evolution
includes \(\cosh(gt)\) and \(\sinh(gt)\) terms.



\begin{thebibliography}{99}

\bibitem{Pezze2018}
L. Pezz\`e, A. Smerzi, M. K. Oberthaler, R. Schmied, and P. Treutlein,
\textit{Quantum metrology with nonclassical states of atomic ensembles},
Rev. Mod. Phys. \textbf{90}, 035005 (2018).

\bibitem{Baum1985}
J. Baum, M. Munowitz, A. N. Garroway, and A. Pines,
\textit{Multiple-quantum dynamics in solid state NMR},
J. Chem. Phys. \textbf{83}, 2015 (1985).

\bibitem{Perelomov1986}
A. Perelomov,
\textit{Generalized Coherent States and Their Applications}
(Springer-Verlag, Berlin, 1986).

\bibitem{Ernst1990}
R. R. Ernst, G. Bodenhausen, and A. Wokaun,
\textit{Principles of Nuclear Magnetic Resonance in One and Two Dimensions}
(Clarendon Press, Oxford, 1990).

\bibitem{GerryKnight2004}
C. C. Gerry and P. L. Knight,
\textit{Introductory Quantum Optics}
(Cambridge University Press, Cambridge, 2004).

\bibitem{Breuer2002}
H.-P. Breuer and F. Petruccione,
\textit{The Theory of Open Quantum Systems}
(Oxford University Press, Oxford, 2002).

\bibitem{Bloembergen1948}
N. Bloembergen, E. M. Purcell, and R. V. Pound,
\textit{Relaxation Effects in Nuclear Magnetic Resonance Absorption},
Phys. Rev. \textbf{73}, 679 (1948).

\bibitem{Abragam1961}
A. Abragam,
\textit{The Principles of Nuclear Magnetism}
(Oxford University Press, Oxford, 1961).

\bibitem{Spohn1978}
H. Spohn,
\textit{Entropy production for quantum dynamical semigroups},
J. Math. Phys. \textbf{19}, 1227 (1978).

\bibitem{Ogawa1990}
S. Ogawa, T. M. Lee, A. R. Kay, and D. W. Tank,
\textit{Brain magnetic resonance imaging with contrast dependent on blood oxygenation},
Proc. Natl. Acad. Sci. U.S.A. \textbf{87}, 9868 (1990).

\bibitem{Henkelman2001}
R. M. Henkelman, G. J. Stanisz, and S. J. Graham,
\textit{Magnetization transfer in MRI: a review},
NMR Biomed. \textbf{14}, 57 (2001).

\bibitem{Kerskens2022}
C. M. Kerskens and D. P\'erez,
\textit{Experimental indications of non-classical brain functions},
J. Phys. Commun. \textbf{6}, 105001 (2022).

\bibitem{Bhatia2019}
R. Bhatia, T. Jain, and Y. Lim,
\textit{On the Bures-Wasserstein distance between positive definite matrices},
Expositiones Mathematicae \textbf{37}, 165 (2019).

\bibitem{Deffner2017}
S. Deffner and S. Campbell,
\textit{Quantum Thermodynamics: An introduction to the thermodynamics of quantum information}
(Morgan \& Claypool Publishers, 2017).

\bibitem{Paper3}
C. Kerskens,
\textit{Evidence for Bures--Wasserstein Boundary Dynamics in the Living Human Brain},
arXiv:2505.22680 [q-bio.NC] (2026).
\url{https://arxiv.org/abs/2505.22680}

\end{thebibliography}
\end{document}